\documentclass[11pt]{article}

\usepackage[margin=1in]{geometry}
\usepackage{amsmath,amssymb,amsthm,amsxtra}
\usepackage{mathrsfs, mathtools}
\usepackage{bbm}
\usepackage{bm}
\usepackage[shortlabels]{enumitem}
\usepackage{microtype}
\usepackage[T1]{fontenc}
\usepackage{xcolor}
\usepackage{comment}
\usepackage{algorithm}
\usepackage{algpseudocode}
\usepackage{float}
\usepackage{hyperref}
\usepackage{graphicx}
\hypersetup{hidelinks}

\numberwithin{equation}{section}

\newtheorem{theorem}{Theorem}[section]
\newtheorem{definition}[theorem]{Definition}
\newtheorem{lemma}[theorem]{Lemma}

\newtheorem{claim}[theorem]{Claim}

\newtheorem{example}[theorem]{Example}

\newtheorem{observation}[theorem]{Observation}

\newcommand{\per}{\operatorname{per}}

\makeatletter
\DeclareRobustCommand\Equiv{\mathrel{%
  \mathchoice
    {\Equiv@\textfont\displaystyle{.45}}
    {\Equiv@\textfont\textstyle{.45}}
    {\Equiv@\scriptfont\scriptstyle{.5}}
    {\Equiv@\scriptscriptfont\scriptscriptstyle{.55}}
}}
\newcommand{\Equiv@}[3]{%
  \rlap{\raisebox{#3\fontdimen5#12}{$\m@th#2 = $}}%
  \raisebox{-#3\fontdimen5#12}{$\m@th#2 = $}%
}
\makeatother

\newcommand{\bitm}{\begin{itemize}[leftmargin=*]}
\newcommand{\eitm}{\end{itemize}}

\newcommand{\benm}{\begin{enumerate}[leftmargin=*]}
\newcommand{\eenm}{\end{enumerate}}

\definecolor{forestgreen}{rgb}{0.13, 0.55, 0.13}

\title{Structural Corrections to the Bethe Approximation of the Permanent}

\author{
Ijay Narang\thanks{School of Computer Science, Georgia Institute of Technology, \texttt{inarang3@gatech.edu}}
\and
Will Perkins\thanks{School of Computer Science, Georgia Institute of Technology, \texttt{wperkins@gatech.edu}. Supported in part by NSF grant CCF-2309708.}
}

\date{}

\begin{document}
\maketitle

\begin{abstract}
We study deterministic approximation algorithms for the permanent of a nonnegative matrix through the Bethe permanent, an approximation computable in polynomial time. The tight analysis of Anari and Rezaei gives a universal comparison between the permanent and the Bethe permanent within a factor \((\sqrt 2)^n\). The simple example of the unweighted \(4\)-cycle \(C_4\) (or a union of disjoint \(C_4\)'s) shows that this bound is tight. 

We show that such \(4\)-cycle obstructions can be identified and exploited algorithmically. Given a Bethe optimizer, our algorithm identifies nearly isolated weighted \(2\times2\) blocks and peels off a vertex-disjoint family of them. If the total weighted correction is large, we can improve the Bethe approximation; if it is small, we show that the Bethe permanent is within a factor of \((\sqrt2 - \varepsilon)^n\) of the truth. Combining these facts, we obtain a deterministic polynomial time \((\sqrt2-\varepsilon)^n\)-approximation algorithm for the permanent of an arbitrary nonnegative \(n\times n\) matrix, where \(\varepsilon>0\) is some absolute constant.
\end{abstract}

\thispagestyle{empty}

\newpage
\setcounter{page}{1}

\section{Introduction}
Given an $n \times n$ matrix $A$, its permanent is defined as
$\per(A) = \sum_{\sigma\in \mathcal{S}_n}\prod_{i=1}^n A_{i\sigma(i)}.$ When $A$ is a $0/1$ matrix, one can view $A$ as the bi-adjacency matrix of a bipartite graph $G$, and accordingly $\per(A)$ is the number of perfect matchings of $G$. For a general nonnegative matrix, the same expression is the  sum of edge-weighted perfect matchings in the support graph of $A$.

The permanent plays a central role in complexity theory, lying at the heart of questions about tractability and the power of randomness in computation. Computing the permanent of a matrix $A$ exactly is $\#P$-hard \cite{valiant1979complexity}, while computing the determinant of $A$ has a polynomial-time algorithm. Furthermore, even deciding the sign of $\per(A)$ is difficult \cite{aaronson2011computational, grier2016new}, motivating the restricted case of nonnegative matrices.

On the randomized side, a breakthrough approach to approximately sampling perfect matchings via Markov chains was initiated by Broder~\cite{Bro86} and culminated in an FPRAS (Fully Polynomial Randomized Approximation Scheme) for the permanent of nonnegative matrices \cite{jerrum1989approximating,jerrum2004polynomial}. That is, for every \(\varepsilon>0\), there is a randomized algorithm running in time polynomial in \(n\) and \(1/\varepsilon\) which outputs, with  probability at least $2/3$, a \((1+\varepsilon)\)-multiplicative approximation to \(\per(A)\).

Constructing an algorithm without randomness, however, remains elusive. No fully polynomial-time approximation scheme (FPTAS) is known, even for \(0/1\) matrices. Consequently, polynomial-time deterministic algorithms for the permanent are usually studied through their exponential approximation rate. We say that a deterministic algorithm gives a \(c^n\)-approximation to the permanent if, on input \(A\), it outputs a quantity \(f(A)\) satisfying
\begin{equation} \label{eq:one_sided_exp_across}
    f(A) \le \per(A) \le c^n f(A).
\end{equation}

The classical matrix-scaling approach of Linial, Samorodnitsky, and Wigderson gives an \(e^n\)-approximation~\cite{linial1998deterministic};  Gurvits and Samorodnitsky~\cite{GurSam14} then gave a $1.9022^n$-approximation.  The sharpest known deterministic polynomial-time guarantee is the \((\sqrt{2})^n\)-approximation obtained by Anari and Rezaei via  the Bethe permanent \cite{anari2019tightanalysisbetheapproximation}.  Improved guarantees are known in certain special cases, such as constant-degree expanders \cite{gamarnik2007deterministicapproximationalgorithmcomputing}. 

The main contribution of this paper is a polynomial-time deterministic algorithm for approximating the permanent of a nonnegative matrix with a better approximation ratio (by an exponential factor) than the Bethe permanent. Our main result is the following:

\begin{theorem} \label{thm:main}
There exists a constant $\varepsilon>0$ and a deterministic, polynomial-time algorithm that, given a matrix $A \in \mathbb R_{\ge0}^{n \times n}$, produces a $(\sqrt2 - \varepsilon)^n$ multiplicative approximation for $\per(A)$.
\end{theorem}

Our algorithm and proof are based on a careful post-processing and analysis of the worst-case behaviour of the Bethe permanent. Therefore, before we discuss our techniques and outline our proofs, we review the Bethe permanent.

\subsection{The Bethe Permanent}

Given a nonnegative matrix $A$ with $\per(A)>0$, define the probability measure $\mu_A$ on perfect matchings by
\begin{equation}
    \mu_A(\sigma) \propto \prod_{i=1}^n A_{i\sigma(i)} .
\end{equation}
The normalizing constant of this probability distribution is exactly $\per(A)$. The distribution  $\mu_A$ induces marginals on each edge: the probability that edge is included in the perfect matching drawn from $\mu_A$. We collect these marginals in the matrix $P \in \mathbb{R}^{n \times n}$, where
\begin{equation}
    P_{ij}=\Pr_{\mu_A}\{\sigma(i)=j\}.
\end{equation}
We refer to the above matrix as the matrix of \emph{true marginals}. Since every perfect matching uses exactly one edge incident to each vertex, the sum marginals of edges incident to any given vertex is $1$. Therefore $P$ lies in the support-restricted perfect matching polytope
\begin{equation}
\mathcal P(A)
    :=
    \left\{
        X\in B_n:
        X_{ij}=0\ \text{ whenever }\ A_{ij}=0
    \right\}.
\end{equation}
Here $B_n$ denotes the Birkhoff polytope of doubly stochastic matrices.

By the standard self-reducibility equivalence between approximate sampling and approximate counting, sufficiently accurate access to the true marginal matrix $P$ would yield an FPTAS for \(\per(A)\). Moreover, one possible route to computing $P$ is to express it as an optimization problem via the Gibbs variational principle. However, the resulting objective includes the entropy of distributions on perfect matchings with prescribed edge marginals and is computationally intractable in general.

This motivates replacing the exact Gibbs variational problem with the Bethe free energy, a variational approximation originally arising in statistical physics and important in statistics and machine learning as well. For additional information on the general Bethe free energy, see for example, \cite{YFW05,mezard2009information,CY13}. Specialized to the perfect matching model, this relaxation gives the Bethe permanent.

For a nonnegative matrix \(A\in\mathbb R_{\ge0}^{n\times n}\) and a doubly stochastic matrix \(P\in B_n\), define
\[
\beta(A,P)
:=
\sum_{i,j}
\left[
P_{ij}\log\frac{A_{ij}}{P_{ij}}
+
(1-P_{ij})\log(1-P_{ij})
\right],
\]
with the conventions \(0\log0=0\) and \(0\log(0/0)=0\). If \(P_{ij}>0\) while \(A_{ij}=0\), the corresponding summand is interpreted as \(-\infty\). The Bethe permanent of \(A\) is then defined as
\begin{equation}
\per_B(A)
:=
\exp\left(\max_{P\in B_n}\beta(A,P)\right).
\end{equation}
Importantly, the Bethe permanent is computable (to arbitrary accuracy) in polynomial time \cite{vontobel2013bethe}. Additionally, we will write
\[
    \phi(t)=-t\log t+(1-t)\log(1-t),
    \qquad
    \Phi(P)=\sum_{i,j}\phi(P_{ij}),
    \qquad
    \mathcal{E}_A(P)=\sum_{i,j}P_{ij}\log A_{ij}.
\]
so that
\[
    \beta(A,P)=\mathcal E_A(P)+\Phi(P)
\]
whenever $P$ is supported on the positive entries of $A$. It is also useful to consider the Bethe entropy term on a vertex-by-vertex level. To that end, we define $\psi : \Delta_m \rightarrow \mathbb{R}$ by $\psi(p) = \sum_i \phi(p_{i})$. A useful property of $\psi$ is that it is concave over the probability simplex.

\begin{lemma}[\cite{vontobel2013bethe}]
\label{clm:psi-concave-simplex}
The function \(\psi_d\) is concave on \(\Delta_d\). Consequently, for every \(p,q\in\Delta_d\),
\[
\psi_d\!\left(\frac{p+q}{2}\right)
-\frac12\psi_d(p)
-\frac12\psi_d(q)
\ge 0.
\]
\end{lemma}

Note that Lemma~\ref{clm:psi-concave-simplex} implies that the Bethe objective is concave on its feasible region. Indeed, \(\mathcal E_A(P)=\sum_{i,j}P_{ij}\log A_{ij}\) is linear in \(P\), while each row \(P_{i,*}\in\Delta_n\), so \(P\mapsto \psi(P_{i,*})\) is concave, giving that \(\Phi\), and therefore \(\beta(A,\cdot)\), is concave on $\mathcal{P}(A)$. Thus maximizing \(\beta\) is equivalent to minimizing the convex function \(-\beta\), so \(\per_B(A)\) is computable to polynomial precision by standard convex optimization methods \cite{vontobel2013bethe,BoydVandenberghe2004}. Because of this efficient computability, much work has gone into understanding the relationship between the permanent and Bethe permanent. In the following theorem, the lower bound is a result of \cite{schrijver1998counting, gurvits2011unleashing} and the upper bound is from \cite{anari2019tightanalysisbetheapproximation}.

\begin{theorem} \label{lem:bethe_perm_known}
For any nonnegative matrix $A \in \mathbb{R}_{\ge0}^{n \times n}$, define $P \in \mathbb{R}_{\ge0}^{n \times n}$ to be the matrix of marginals induced by the measure $\mu_A(\sigma) \propto \prod_{i=1}^n A_{i\sigma(i)}$. Then
    \begin{equation}
        \per_B(A) \leq \per(A) \leq (\sqrt{2})^n\per_B(A; P) \leq (\sqrt{2})^n \per_B(A),
    \end{equation}
where $\per_B(A; P) = \exp(\beta(A, P))$, the Bethe permanent evaluated at $P$.
\end{theorem}

The lower bound was originally proved via an inequality by Schrijver~\cite{schrijver1998counting,gurvits2011unleashing} and admits alternative proofs and generalizations via lifts~\cite{Csi14} and stable polynomials~\cite{AO17}. In contrast, the upper bound proof of \cite{anari2019tightanalysisbetheapproximation} follows an entropy comparison framework similar to that of \cite{Rad97} for Bregman's Inequality \cite{Bre73}.

Additionally, we note that both the upper and lower bounds are tight: the lower bound is tight by the identity matrix, and the upper bound is tight by the following extremal example.

\begin{example} Consider the unweighted cycle $C_4$. Then, its set of Bethe optimal solutions is given by the one-parameter family
\[
\left\{
\begin{pmatrix}
p & 1-p\\
1-p & p
\end{pmatrix}
\;\middle|\;
p\in[0,1]
\right\}.
\]
A direct computation gives that $\frac{\per(G) }{\per_B(G)} = 2$, and thus a set of $n/2$ disjoint $C_4$'s gives an approximation ratio of $(\sqrt{2})^n.$
\end{example}

A natural question is whether a disjoint set of $C_4$'s is the unique worst-case example, and after accounting for edges with marginals equal to $0$ or $1$, it is. However, one can construct many examples with ratio $\frac{\operatorname{per}(A) }{\operatorname{per}_B(A)} = (\sqrt{2} -o(1))^n$ that can potentially look very different.

\begin{figure}[H]
    \centering
    \includegraphics[width=0.45\textwidth]{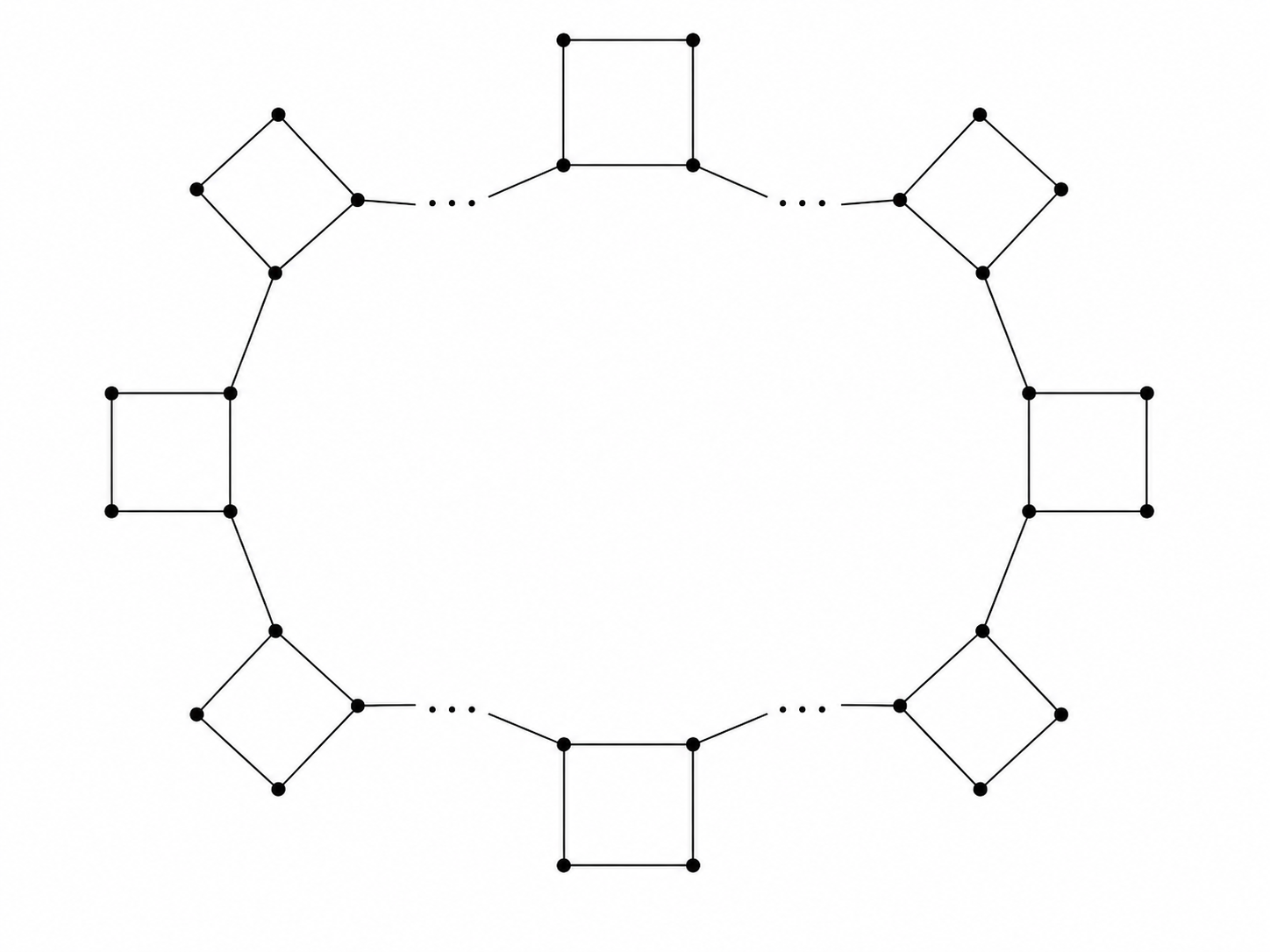}
    \caption{As the outer cycle length tends to infinity, the edges inside each displayed weighted \(2\times2\) gadget have true marginals tending to \(1/2\), while the connector edges between successive gadgets have marginals tending to \(0\). Thus the marginal vector becomes arbitrarily close to that of disconnected \(C_4\)-like components, even though the support graph is connected.}
    \label{fig:glued_c4_example}
\end{figure}

Motivated by such examples, a core idea of our proof is to show that any nonnegative matrix $A$ which attains a near worst-case approximation ratio must behave like a collection of disjoint weighted $2\times2$ blocks whose two local orientation products are nearly balanced. This naturally lends itself to the following dichotomy: either the Bethe permanent performs much better than the worst-case, or it has weighted $C_4$-like behavior. We are able to correct this behavior via a post-processing algorithm that peels off weighted $2\times2$ components.

\subsection{Outline of the Proof}

We begin by observing that we can exactly correct weighted $4$-cycles. Indeed, consider
\[
C=
\begin{pmatrix}
a&b\\ c&d
\end{pmatrix}.
\]
The two local matching orientations have weights $U=ad$ and $V=bc.$ A direct computation gives that $ \operatorname{per}_B(C)=\max\{U,V\}$ and $  \operatorname{per}(C)=U+V$. Thus, the exact correction factor is
\begin{equation} \label{eq:weighted_local_correction_intro}
    s_C
    =
    \frac{\operatorname{per}(C)}{\operatorname{per}_B(C)}
    =
    1+\frac{\min\{U,V\}}{\max\{U,V\}}.
\end{equation}

In a general matrix, however, a \(2\times2\) block need not be isolated. Perfect matchings may enter
or leave the block through boundary edges. We therefore look for blocks which are nearly isolated. This leads to the following definition.

\begin{definition} \label{def:weighted_blocks}
Let $A\in\mathbb R_{\ge0}^{n\times n}$ and let $P\in\mathcal P(A)$. For a block $Q=(\{\ell_1,\ell_2\},\{r_1,r_2\}),$ with all four entries $A_{\ell_i r_j}$ positive, define
\[
U_Q:=A_{\ell_1r_1}A_{\ell_2r_2},
\qquad
V_Q:=A_{\ell_1r_2}A_{\ell_2r_1},
\]
\[
\theta_Q:=\frac{\min\{U_Q,V_Q\}}{\max\{U_Q,V_Q\}},
\qquad
s_Q:=1+\theta_Q,
\qquad
g_Q:=\log s_Q.
\]
The $P$-leakage of $Q$ is
\[
\operatorname{leak}_{P}(Q)
:=
\sum_{i=1}^2\sum_{r\notin\{r_1,r_2\}}P_{\ell_i r}
+
\sum_{j=1}^2\sum_{\ell\notin\{\ell_1,\ell_2\}}P_{\ell r_j}.
\]
We say that $Q$ is $\lambda$-admissible for $(A,P)$ if $\operatorname{leak}_{P}(Q)\le \lambda$.
\end{definition}

Given $A$ and $P$, let $\Gamma_\lambda(A,P)$ denote the total score produced by the greedy algorithm that orders all $\lambda$-admissible blocks by decreasing $g_Q$ and selects a block whenever it is vertex-disjoint from all previously selected blocks. That is,
\begin{equation} \label{eq:greedy_collection}
    \Gamma_\lambda(A,P)
    =
    \sum_{Q\in\mathcal C_{\rm gr}}\log s_Q,
\end{equation}
where $\mathcal C_{\rm gr}$ is the greedy vertex-disjoint collection.

The Bethe approximation performs worst when it handles many high-score weighted $2\times2$ obstructions. We show that when the Bethe optimizer has large total obstruction score, there exists an algorithm that corrects these obstructions.

\begin{theorem} \label{thm:dense-c4}
Fix $\tau>0$. There exist $\lambda_0=\lambda_0(\tau)>0$ and $\varepsilon=\varepsilon(\tau)>0$ such that the following holds for every $0<\lambda\le\lambda_0$.

Let $A\in\mathbb R_{\ge0}^{n\times n}$ satisfy $\operatorname{per}(A)>0$, and let $P^*$ be a Bethe optimizer. If
\[
\Gamma_\lambda(A,P^*)\ge \tau n,
\]
then Algorithm~\ref{alg:dense-c4-correction} outputs a number $L_{\rm out}$ satisfying
\[
L_{\rm out}\le \operatorname{per}(A)
\le (\sqrt2-\varepsilon)^n L_{\rm out}.
\]
\end{theorem}

The underlying algorithmic idea is that when a solution to the Bethe variational problem has many weighted $2\times2$ defects, we can construct an auxiliary weighted matrix which provably has significantly smaller permanent, but similar Bethe value. By analyzing the structural differences between these matrices, we are able to produce a better approximation of the permanent.

To handle the case of small total obstruction score, we show that then the Bethe approximation ratio is already significantly improved. This is given by the following result.

\begin{theorem} \label{thm:sparse-c4}
There exist constants $\lambda>0$, $\tau>0$, and $\varepsilon>0$ such that the following holds.

Let $A\in\mathbb R_{\ge0}^{n\times n}$ satisfy $\operatorname{per}(A)>0$, and let $P^*$ be a Bethe optimizer. If
\[
\Gamma_\lambda(A,P^*)<\tau n,
\]
then
\[
\per_B(A) \leq \operatorname{per}(A)
\le (\sqrt2-\varepsilon)^n\operatorname{per}_B(A).
\]
\end{theorem}

Our proof of Theorem~\ref{thm:sparse-c4} is a stability refinement of the
row-exposure framework of Anari and Rezaei~\cite{anari2019tightanalysisbetheapproximation}.
For the true marginal matrix \(P\), the energy term cancels exactly:
\[
    \log\operatorname{per}(A)-\beta(A,P)
    =
    H(\mu_A)-\Phi(P).
\]
Thus the comparison problem becomes purely entropic. The Anari--Rezaei
one-row slack terms are nonnegative and vanish only at two-point extremizers
with masses \((1/2,1/2)\). Thus a near-extremal ratio forces almost
all rows of the true marginal matrix to be close to such two-point vectors.

This row-wise conclusion alone does not imply \(C_4\)-structure: the near-\(1/2\)
edges could form longer paths or cycles. We rule out this possibility through
a block-exposure entropy argument, which shows that non-\(C_4\) components incur
a linear entropy loss. Hence near-extremal true marginals must organize into
\(C_4\)-components. We then show that most of these true-tight components have
small leakage with respect to the Bethe optimizer \(P^*\), and that their two
weighted orientations are nearly balanced. Therefore, if
\(\Gamma_\lambda(A,P^*)\) is small, the Bethe ratio is already bounded away from
\((\sqrt2)^n\).

\paragraph{Independent work}

After finishing and submitting this paper, we learned of independent work of Anari~\cite{anari2026beyond} with a similar main result, but different algorithm.

\paragraph{Organization.}

The remainder of this paper is organized as follows. Section \ref{sec:the_algorithm} contains our algorithm, and Sections \ref{sec:dense-case} and \ref{sec:sparse-case} contain the proofs of Theorems \ref{thm:dense-c4} and \ref{thm:sparse-c4} respectively. These two theorems together immediately imply Theorem \ref{thm:main}.

\paragraph{AI Use Statement}

AI (ChatGPT 5.5 Plus) was used to check the mathematical proofs as well as spelling and grammar. It also suggested ideas that helped simplify Lemma \ref{lem:near_extremal_true_c4_structure_quantitative}.

\section{The Algorithm} \label{sec:the_algorithm}

We now give the deterministic approximation algorithm. The algorithm first checks whether the support of \(A\) contains a perfect matching. If not, then \(\operatorname{per}(A)=0\) and the algorithm returns \(0\). Otherwise, it
computes a Bethe optimizer \(P^*\), identifies a greedy vertex-disjoint family
of low-leakage \(2\times2\) blocks with respect to \(P^*\), and either returns
the uncorrected Bethe permanent or applies the correction subroutine described
below.
The constants \(\lambda\) and \(\tau\) are fixed absolute constants chosen  so that the two guarantees in Theorems~\ref{thm:dense-c4}
and~\ref{thm:sparse-c4} apply. For simplicity we describe the algorithms as using exact Bethe optimizers and exact Bethe
values. Standard convex optimization computes the required quantities to inverse-polynomial
additive accuracy, and the resulting error can be absorbed by decreasing the final constant
\(\varepsilon\).

\begin{algorithm}[H]
\caption{Deterministic weighted Bethe correction algorithm}
\label{alg:bethe-c4-correction}
\small
\begin{algorithmic}[1]
\Require A nonnegative matrix \(A\in\mathbb R_{\ge 0}^{n\times n}\).
\Ensure A lower estimate \(L_{\mathrm{out}}\le \operatorname{per}(A)\).

\State If the support of \(A\) has no perfect matching, return \(0\).

\State  Use the fixed constants \(\lambda,\tau>0\) small enough for Theorems \ref{thm:dense-c4} and \ref{thm:sparse-c4}.

\State Compute
\[
    b^* := \max_{X\in \mathcal P(A)} \beta(A,X),
\]
and let \(P^*\in\arg\max_{X\in\mathcal P(A)}\beta(A,X)\).

\State Compute the greedy vertex-disjoint collection $\mathcal C$ of $\lambda$-admissible blocks from Definition \ref{def:weighted_blocks} and Equation \eqref{eq:greedy_collection}, ordered by decreasing score $g_Q=\log s_Q$.
\State Set $S:=\sum_{Q\in\mathcal C}\log s_Q=\Gamma_\lambda(A,P^*)$.

\If{$S <\tau n$}
    \State \Return $L_{\mathrm{out}}:=\operatorname{per}_B(A)=\exp(b^*)$.
\Else
    \State Run Algorithm \ref{alg:dense-c4-correction} on input $(A,P^*,\mathcal C)$.
    \State \Return its output $L_{\mathrm{out}}$.
\EndIf
\end{algorithmic}
\end{algorithm}

Note that Algorithm~\ref{alg:bethe-c4-correction} runs in polynomial time. It  requires a check for the existence of a perfect matching, a computation of the Bethe permanent, the additional
step of computing \(\Gamma_\lambda(A,P^*)\), and a search over \(O(n^4)\)
possible \(2\times2\) blocks.

Algorithm~\ref{alg:dense-c4-correction} is designed to repair the places where the Bethe approximation loses weighted factors due to local ambiguity. In a weighted $2\times2$ block $C=(\{\ell_1,\ell_2\},\{r_1,r_2\})$, the true matching model has two local orientations: the parallel orientation $\{(\ell_1,r_1),(\ell_2,r_2)\}$ with weight $U_C$, and the cross orientation $\{(\ell_1,r_2),(\ell_2,r_1)\}$ with weight $V_C$. The weighted Bethe correction factor is $s_C=1+\min\{U_C,V_C\}/\max\{U_C,V_C\}$.

The algorithm makes this correction explicit.  For each selected \(2\times2\) block, the algorithm keeps the heavier local orientation intact and deletes one edge from the lighter orientation, choosing the deleted edge according to the \(P^*\)-masses.  The resulting matrix $A^\circ$ is then modified on boundary edges: edges leaving a selected left vertex are divided by the local correction factor $s_C$. After this is done across all selected blocks, the resulting estimate approximates the weighted number of perfect matchings in the modified matrix by the Bethe permanent and then multiplies by the product of the local correction factors.

Below is the explicit pseudocode for the correction procedure.

\begin{algorithm}[H]
\caption{Bethe guided weighted $C_4$-correction subroutine}
\label{alg:dense-c4-correction}
\small
\begin{algorithmic}[1]
\Require A nonnegative matrix $A$, a Bethe optimizer $P^*$, and a vertex-disjoint collection $\mathcal C$ of $\lambda$-admissible blocks.
\Ensure  A lower estimate \(L_{\mathrm{out}}\le \operatorname{per}(A)\).

\State Initialize $A^\circ\gets A$.
\For{each $C\in\mathcal C$}
    \State Write $L(C)=\{\ell_1,\ell_2\}$ and $R(C)=\{r_1,r_2\}$.
    \State Set
    \[
        U_C=A_{\ell_1r_1}A_{\ell_2r_2},
        \qquad
        V_C=A_{\ell_1r_2}A_{\ell_2r_1},
        \qquad
        s_C=1+\frac{\min\{U_C,V_C\}}{\max\{U_C,V_C\}}.
    \]
    \If{$U_C\ge V_C$}
        \State Let $b:=P^*_{\ell_1r_2}$ and $c:=P^*_{\ell_2r_1}$.
        \If{$b\le c$}
            \State Delete $(\ell_1,r_2)$ from $A^\circ$.
        \Else
            \State Delete $(\ell_2,r_1)$ from $A^\circ$.
        \EndIf
    \Else
        \State Let $a:=P^*_{\ell_1r_1}$ and $d:=P^*_{\ell_2r_2}$.
        \If{$a\le d$}
            \State Delete $(\ell_1,r_1)$ from $A^\circ$.
        \Else
            \State Delete $(\ell_2,r_2)$ from $A^\circ$.
        \EndIf
    \EndIf
\EndFor

\State Define a weighted matrix $H$ on the support of $A^\circ$ as follows.
\For{each positive entry $e=(\ell,r)$ of $A^\circ$}
    \If{$\ell\in L(C)$ for some $C\in\mathcal C$, and $r\notin R(C)$}
        \State Set $H_e:=A^\circ_e/s_C$.
    \Else
        \State Set $H_e:=A^\circ_e$.
    \EndIf
\EndFor

\State Compute the weighted Bethe permanent $\operatorname{per}_B(H)$.
\State \Return
\[
L_{\mathrm{out}}:=\left(\prod_{C\in\mathcal C}s_C\right)\operatorname{per}_B(H).
\]
\end{algorithmic}
\end{algorithm}

It is straightforward to see that Algorithm~\ref{alg:dense-c4-correction} runs in polynomial time. To verify that the worst case approximation ratio is indeed $(\sqrt2 - \varepsilon)^n$, we must show that when $\Gamma_\lambda(A,P^*)<\tau n$, the quantity $\frac{\per(A)}{\per_B(A)}$ is at most $(\sqrt2 - \varepsilon)^n$, and that when $\Gamma_\lambda(A,P^*)\ge \tau n$, the correction procedure leads to the claimed correction. These two properties follow from Theorems \ref{thm:sparse-c4} and \ref{thm:dense-c4} respectively.

\section{Proof of Theorem \ref{thm:dense-c4}}
\label{sec:dense-case}

\begin{proof}
The algorithmic procedure which achieves this guarantee is exactly Algorithm~\ref{alg:dense-c4-correction}. Our proof  proceeds by analyzing its steps and the resulting approximation ratio.

Let $P^*$ be the Bethe optimizer under consideration. Let $\mathcal C$ be the collection of blocks selected by the greedy procedure in Algorithm~\ref{alg:bethe-c4-correction}, and set
\[
    S:=\sum_{C\in\mathcal C}\log s_C=\Gamma_\lambda(A,P^*).
\]
By assumption, $S\ge \tau n$. Since the components in $\mathcal C$ are vertex-disjoint, and each uses two left vertices, we also have $|\mathcal C|\le n/2$.

Recall that Algorithm~\ref{alg:dense-c4-correction} constructs a matrix $A^\circ$ by deleting one internal edge from each component $C\in\mathcal C$. We will show that the Bethe free energy of $A^\circ$ is roughly similar to that of $A$ by constructing a feasible point $P^\circ\in\mathcal P(A^\circ)$. Fix a component $C\in\mathcal C$, and write its vertices as
\[
L(C)=\{\ell_1,\ell_2\},
\qquad
R(C)=\{r_1,r_2\}.
\]
For notational convenience write
\[
    a=P^*_{\ell_1r_1},\quad
    b=P^*_{\ell_1r_2},\quad
    c=P^*_{\ell_2r_1},\quad
    d=P^*_{\ell_2r_2}.
\]
Since $C$ is $\lambda$-admissible, the total $P^*$-leakage from the two rows and two columns of $C$ is at most $\lambda$. In particular, the two internal row sums and the two internal column sums all lie in $[1-\lambda,1]$.

Suppose first that $U_C\ge V_C$, so that the parallel orientation is at least as heavy as the cross orientation. If $b\le c$, Algorithm~\ref{alg:dense-c4-correction} deletes $(\ell_1,r_2)$ and we define
\[
    a^\circ=a+b,
    \qquad
    b^\circ=0,
    \qquad
    c^\circ=c-b,
    \qquad
    d^\circ=d+b.
\]
Equivalently,
\[
P^\circ_{\ell_1r_1}=a+b,
\quad
P^\circ_{\ell_1r_2}=0,
\quad
P^\circ_{\ell_2r_1}=c-b,
\quad
P^\circ_{\ell_2r_2}=d+b.
\]
All edges outside selected blocks keep their $P^*$ value. Observe that this solution preserves the two row sums and the two column sums. Furthermore, the non-negativity constraints are met because $c^\circ=c-b\ge0$, and $a^\circ\le1$, $d^\circ\le1$ by the original row and column constraints for $P^*$. Therefore the new values remain feasible, and the deleted edge receives value $0$. The case $c<b$ is symmetric. If $U_C<V_C$, we perform the same alternating-cycle transformation with the roles of the two orientations reversed. Since the components of $\mathcal C$ are vertex-disjoint, these local transformations are compatible and yield a feasible point $P^\circ\in\mathcal P(A^\circ)$.

We now compare the Bethe objective of $P^\circ$ with that of $P^*$. The energy term does not decrease. Indeed, in the case $U_C\ge V_C$ and $b\le c$, the internal energy change is
\[
    b\log\frac{A_{\ell_1r_1}A_{\ell_2r_2}}{A_{\ell_1r_2}A_{\ell_2r_1}}
    =
    b\log\frac{U_C}{V_C}
    \ge0.
\]

It remains to control the Bethe entropy term. We use the following elementary estimate.

\begin{claim} \label{clm:near_closed_entropy}
There is an absolute constant $C>0$ such that the following holds for all sufficiently small $\lambda>0$. Let
\[
B=
\begin{pmatrix}
x_{11}&x_{12}\\x_{21}&x_{22}
\end{pmatrix}
\]
be a nonnegative $2\times2$ matrix whose two row sums and two column sums all lie in $[1-\lambda,1]$. Then
\[
    \left|\sum_{i,j=1}^2\phi(x_{ij})\right|
    \le C\lambda\log(1/\lambda).
\]
Consequently, if $B'$ is another nonnegative $2\times2$ matrix with the same row and column sums, then
\[
    \sum_{i,j=1}^2\phi(B'_{ij})
    \ge
    \sum_{i,j=1}^2\phi(B_{ij})
    -C\lambda\log(1/\lambda).
\]
\end{claim}

\begin{proof}[Proof of Claim \ref{clm:near_closed_entropy}]
Set $x=x_{11}$. Since the first row sum is in $[1-\lambda,1]$, the entry $x_{12}$ differs from $1-x$ by at most $\lambda$. Since the first column sum is in $[1-\lambda,1]$, the entry $x_{21}$ differs from $1-x$ by at most $\lambda$. Finally, since the second row sum and first column sum both lie in $[1-\lambda,1]$, the entry $x_{22}$ differs from $x$ by at most $\lambda$.

Note that the function $\phi$ satisfies 
\[
    |\phi(u)-\phi(v)|\le C\lambda\log(1/\lambda)
    \qquad\text{whenever } |u-v|\le \lambda,
\]
for all $u,v\in[0,1]$ and all sufficiently small $\lambda$. Along the exactly closed face $(x,1-x,1-x,x)$ we have
\[
    \phi(x)+\phi(1-x)+\phi(1-x)+\phi(x)=0.
\]
Combining these two facts gives the desired bound. The final assertion follows by applying the same bound to both $B$ and $B'$.
\end{proof}

Applying Claim \ref{clm:near_closed_entropy} to each selected block and using the nondecrease of the energy term gives
\[
\beta(A^\circ,P^\circ)
\ge
\beta(A,P^*)
-
C_1\lambda\log(1/\lambda)n.
\]
By definition, $\beta(A,P^*)=\log\operatorname{per}_B(A)$, and since $P^\circ$ is feasible for $A^\circ$, this implies
\begin{equation} \label{eq:weighted_bethe_after_deletion}
\log\operatorname{per}_B(A^\circ)
\ge
\beta(A^\circ,P^\circ)
\ge
\log\operatorname{per}_B(A)
-
C_1\lambda\log(1/\lambda)n.
\end{equation}

With this fact in hand, we leverage the difference between $A^\circ$ and $A$ in order to explicitly correct the Bethe error on $A$. Let $\mathcal M(A^\circ)$ denote the set of perfect matchings in the support of $A^\circ$. For a matching $M\in\mathcal M(A^\circ)$, say that a component $C\in\mathcal C$ is closed by $M$ if the two left vertices of $C$ are matched by $M$ to the two right vertices of $C$.

On each closed component $C$, the matching $M$ uses the unique internal perfect matching of $C$ which remains in $A^\circ$ and corresponds to the heavier orientation. Since $A$ contains both internal perfect matchings of $C$, we may flip independently on each closed component. The total weighted contribution from the two choices on $C$ is multiplied by exactly $s_C$.

Moreover, two distinct matchings $M,M'\in\mathcal M(A^\circ)$ generate disjoint flip orbits in the support of $A$. Indeed, if two matchings in $A^\circ$ differed by a nontrivial flip on some $C\in\mathcal C$, then one of them would use the internal edge deleted from $A^\circ$, contradicting membership in $\mathcal M(A^\circ)$.

Recall that Algorithm~\ref{alg:dense-c4-correction} defines a weighted matrix $H$ by dividing boundary edges leaving selected left vertices by their local $s_C$. For $M\in\mathcal M(A^\circ)$, let $b_C(M)$ be the number of selected left vertices of $C$ which are matched outside $R(C)$. Then
\[
    w_H(M)=w_A(M)\prod_{C\in\mathcal C}s_C^{-b_C(M)},
\]
where $w_A(M)=\prod_{e\in M}A_e$. If $C$ is closed by $M$, then $b_C(M)=0$. If $C$ is not closed by $M$, then at least one selected left vertex is matched outside $R(C)$, so $b_C(M)\ge1$. Therefore
\[
\left(\prod_{C\in\mathcal C}s_C\right)w_H(M)
\le
w_A(M)\prod_{C\text{ closed by }M}s_C.
\]
The right-hand side is exactly the weighted contribution of the flip orbit generated from $M$. Summing over $M\in\mathcal M(A^\circ)$ gives
\begin{equation} \label{eq:weighted_orbit_lower}
    \left(\prod_{C\in\mathcal C}s_C\right)\operatorname{per}(H)
    \le
    \operatorname{per}(A).
\end{equation}
Since $\operatorname{per}_B(H)\le\operatorname{per}(H)$ by Theorem \ref{lem:bethe_perm_known}, the output of Algorithm~\ref{alg:dense-c4-correction} is indeed a lower bound:
\[
L_{\rm out}
=
\left(\prod_{C\in\mathcal C}s_C\right)\operatorname{per}_B(H)
\le
\operatorname{per}(A).
\]

It remains to compare $\operatorname{per}_B(H)$ with $\operatorname{per}_B(A)$. Since $H$ is obtained from $A^\circ$ by dividing boundary edges of each selected block $C$ by $s_C$, we have
\[
\beta(H,P^\circ)
=
\beta(A^\circ,P^\circ)
-
\sum_{C\in\mathcal C}(\log s_C)
\sum_{\ell\in L(C)}\sum_{r\notin R(C)}P^\circ_{\ell r}.
\]
The transformation from $P^*$ to $P^\circ$ does not change boundary edge values. Since each selected block is $\lambda$-admissible, its selected-left boundary mass is at most $\lambda$. Therefore
\[
\sum_{C\in\mathcal C}(\log s_C)
\sum_{\ell\in L(C)}\sum_{r\notin R(C)}P^\circ_{\ell r}
\le
\lambda\sum_{C\in\mathcal C}\log s_C
\le
\lambda n\log2.
\]
Combining this with \eqref{eq:weighted_bethe_after_deletion} gives
\begin{equation} \label{eq:weighted_bethe_H}
\log \operatorname{per}_B(H)
\ge
\log \operatorname{per}_B(A)-C_2\lambda\log(1/\lambda)n
\end{equation}
for an absolute constant $C_2$. Now, by Equation \eqref{eq:weighted_bethe_H},
\begin{align*}
L_{\mathrm{out}} &=
\left(\prod_{C\in\mathcal C}s_C\right)\operatorname{per}_B(H) \\
&=
\exp(S)\operatorname{per}_B(H) \\
&\ge
\exp(S-C_2\lambda\log(1/\lambda)n)\operatorname{per}_B(A).
\end{align*}
Plugging in the above and using Theorem \ref{lem:bethe_perm_known}, we obtain
\begin{align*}
\frac{\operatorname{per}(A)}{L_{\mathrm{out}}}
&\le
\exp(C_2\lambda\log(1/\lambda)n)\frac{(\sqrt2)^n}{e^S} \\
&\le
\left[\sqrt2\,\exp\{-\tau+C_2\lambda\log(1/\lambda)\}\right]^n.
\end{align*}
as $S\ge \tau n$. Choosing $\lambda_0=\lambda_0(\tau)>0$ small enough so that $C_2\lambda\log(1/\lambda)<\tau/2$ whenever $0<\lambda\le\lambda_0$ gives the result
\[
\frac{\operatorname{per}(A)}{L_{\mathrm{out}}}
\le
\left(\sqrt2 e^{-\tau/2}\right)^n
=
(\sqrt2-\varepsilon)^n
\]
for some $\varepsilon=\varepsilon(\tau)>0$.
\end{proof}

\section{Proof of Theorem \ref{thm:sparse-c4}}
\label{sec:sparse-case}

The proof has three parts. First, we use the row-exposure entropy comparison of Anari and Rezaei~\cite{anari2019tightanalysisbetheapproximation} to show that near equality in the Bethe upper bound forces
almost all true marginal rows to be close to two-point vectors. Second, we prove a block-exposure refinement showing that these two-point rows must organize into
\(C_4\)-components; longer paths and cycles lose a linear amount of entropy. Third, we transfer this true-marginal structure to the Bethe optimizer \(P^*\)
using concavity of the Bethe entropy, and show that the corresponding weighted \(2\times2\) blocks have nearly balanced orientation products.
We will signpost our use of the results of \cite{anari2019tightanalysisbetheapproximation}; we recommend the reader also review their methodology. 

We begin with the following useful definition.

\begin{definition}
Given a doubly stochastic matrix $P \in \mathbb{R}^{n \times n}$, we say a row $i$ is \emph{$\delta$-two point} if
there exist two distinct indices \(j_1,j_2\) such that
\[
    \left|P_{ij_1}-\frac12\right|\le \delta,
    \qquad
    \left|P_{ij_2}-\frac12\right|\le \delta .
\]
\end{definition}

At a high level, our proof proceeds by first showing that if the worst-case ratio of $(\sqrt{2} - o(1))^n$ is achieved, then true marginals must look like two-point rows, and furthermore that they will be organized according to $4$-cycles. Once we have shown this, we show that most of these true-tight $4$-cycles are low-leakage blocks for the Bethe optimizer and have nearly balanced weighted orientation products. This presence of many high-score admissible blocks contradicts the assumption that $\Gamma_\lambda(A,P^*)$ is small.

Before we proceed to the analysis, we note that, although the matrix \(A\) is weighted, the comparison at the true marginal matrix is purely entropic. 

\begin{observation}\label{obs:true-marginal-entropy-gap}
Let \(P\) be the true marginal matrix of \(\mu_A\). Then
\[
    \log \operatorname{per}(A)-\beta(A,P)
    =
    H(\mu_A)-\Phi(P).
\]
Consequently,
\[
    \log \frac{\operatorname{per}(A)}{\operatorname{per}_B(A)}
    \le
    H(\mu_A)-\Phi(P).
\]
\end{observation}
Indeed, under \(\mu_A\), \(\log \mu_A(M)=\log w_A(M)-\log \operatorname{per}(A)\), so
\[
\log\operatorname{per}(A)=H(\mu_A)+\mathbb E_{\mu_A}\log w_A(M).
\]
Since \(\mathbb E_{\mu_A}\log w_A(M)=\sum_e P_e\log A_e=\mathcal E_A(P)\), the identity follows.

Thus, throughout the proof, we will work with the quantity $H(\mu_A)-\Phi(P)$.

We can now begin with the following lemma. The strategy of its proof is similar to the coordinate merging framework utilized in \cite{anari2019tightanalysisbetheapproximation}. We will also utilize the following consequences of their analysis: the row-slack decomposition, the coordinate-merging monotonicity lemma, and the three-dimensional equality cases of the one-row inequality.

\begin{lemma}
\label{lem:true_two_point_quantitative}
Fix $\delta>0$ sufficiently small and $\kappa>0$. There exists a constant $c=c(\delta,\kappa)>0$ such that the following holds.

Let $A\in\mathbb R_{\ge0}^{n\times n}$ have $\per(A)>0$. Let $P$ denote the true marginal matrix of the Gibbs measure $\mu_A$. If at least $\kappa n$ rows of $P$ are not $\delta$-two point, then
\[
    \frac{\operatorname{per}(A)}{\operatorname{per}_B(A)}
    \le
    \bigl(\sqrt2-c\bigr)^n .
\]
Equivalently, if
\[
    \frac{\operatorname{per}(A)}{\operatorname{per}_B(A)}
    >
    \bigl(\sqrt2-c\bigr)^n,
\]
then all but at most $\kappa n$ rows of $P$ are $\delta$-two point.
\end{lemma}

\begin{proof}
Let $\mu=\mu_A$, and let $M\sim\mu$. For $\ell\in L$, write $M(\ell)$ for the right vertex matched to $\ell$. For a probability vector $p=(p_1,\dots,p_m)\in\Delta_m$, define
\[
\Lambda(p)
:=
\log\sqrt2
-
\mathbb E_{\pi}
\left[
\sum_{k=1}^m
p_k
\log\left(\sum_{j:\,j\ge_\pi k}p_j\right)
-
\sum_{k=1}^m
(1-p_k)\log(1-p_k)
\right],
\]
where $\pi$ is a uniformly random ordering of $[m]$, and $j\ge_\pi k$ means that $j=k$ or $j$ appears after $k$ in the ordering $\pi$. All expressions are interpreted with the usual convention $0\log 0=0$.

We use the following consequence of the row-exposure analysis of \cite{anari2019tightanalysisbetheapproximation}. Applied to the true marginal matrix $P$, it gives the row-slack decomposition
\[
    H(M)-\Phi(P)
    \le
    n\log\sqrt2
    -
    \sum_{\ell\in L}\Lambda(P_{\ell,\cdot}).
\]
Moreover, the one-row inequality of \cite{anari2019tightanalysisbetheapproximation} gives $\Lambda(p)\ge 0$ for every probability vector $p$.

We claim that there exists a constant $\gamma(\delta)>0$ such that every probability vector $p$ which is not $\delta$-two point satisfies $\Lambda(p)\ge \gamma(\delta)$. Let $p=(p_1,\dots,p_m)\in\Delta_m$ be a probability vector which is not $\delta$-two point. Zero coordinates do not affect the row expression, so we delete all zero coordinates. Starting from the remaining positive coordinates, repeatedly merge the two currently smallest coordinates until only three coordinates remain. If fewer than three positive coordinates remain, append zero coordinates. Let $H(p)\in\Delta_3$ be the resulting three-coordinate vector.

We now justify that every merge used above is allowed by the coordinate-merging monotonicity lemma of \cite{anari2019tightanalysisbetheapproximation}. Suppose that at some stage there are $r>3$ active coordinates, and let $a\le b$ be the two smallest active masses. Since the total mass is $1$, we have
\[
    a+b\le \frac{2}{r}\le \frac12.
\]
Thus the merging of \cite{anari2019tightanalysisbetheapproximation} (Proof of Lemma 13) implies that these merges cannot increase the one-row slack, and hence
\[
    \Lambda(p)\ge \Lambda(H(p)).
\]
Let
\[
    \mathcal E_3
    :=
    \left\{
        \left(\frac12,\frac12,0\right),
        \left(\frac12,0,\frac12\right),
        \left(0,\frac12,\frac12\right)
    \right\}.
\]
We claim that $H(p)$ is bounded away from $\mathcal E_3$ by a quantity depending only on $\delta$. Write the coordinates of $H(p)$ in increasing order as $q_1\le q_2\le q_3$.

We first record a simple property of the smallest-coordinate merging procedure. If one of the final coordinates is not an original coordinate of $p$, then its mass is at most $2q_1$. Indeed, suppose a final coordinate $w$ was created at its last internal merge by merging two active blocks of masses $a\le b$. At that moment, $a$ and $b$ were the two smallest active masses, so every active block outside this merge had mass at least $b$. If $w=q_1$, then $w\le 2q_1$ is immediate. Otherwise, $q_1$ descends from an active block outside this merge. Since subsequent merges only increase masses, we have $q_1\ge b$. Hence $w=a+b\le 2b\le 2q_1$.

Now fix an absolute constant $a_0>0$ sufficiently small and take $\delta>0$ sufficiently small. Suppose, for contradiction, that $\operatorname{dist}_1(H(p),\mathcal E_3)<a_0\delta$. Since $\mathcal E_3$ is invariant under coordinate permutations, after sorting coordinates increasingly this implies
\[
    \left\|(q_1,q_2,q_3)-\left(0,\frac12,\frac12\right)\right\|_1
    <a_0\delta.
\]
In particular,
\[
    q_1<a_0\delta,
    \qquad
    \left|q_2-\frac12\right|<a_0\delta,
    \qquad
    \left|q_3-\frac12\right|<a_0\delta.
\]
For $a_0$ and $\delta$ chosen sufficiently small, this gives $q_2>2q_1$ and $q_3>2q_1$. From the property proved above, neither $q_2$ nor $q_3$ can be a merged coordinate. Therefore $q_2$ and $q_3$ must both be original coordinates of $p$. But then
\[
    \left|q_2-\frac12\right|<a_0\delta<\delta,
    \qquad
    \left|q_3-\frac12\right|<a_0\delta<\delta,
\]
so $p$ has two original coordinates within $\delta$ of $1/2$. This contradicts the assumption that $p$ is not $\delta$-two point. Therefore $\operatorname{dist}_1(H(p),\mathcal E_3)\ge a_0\delta$. Now note that the function $\Lambda$ is continuous on $\Delta_3$ and moreover, the equality-case form of \cite{anari2019tightanalysisbetheapproximation} gives
\[
    \Lambda(q)=0
    \qquad\Longleftrightarrow\qquad
    q\in\mathcal E_3
    \qquad
    \text{for }q\in\Delta_3.
\]
By compactness,
\[
    \gamma(\delta)
    :=
    \inf
    \left\{
        \Lambda(q):
        q\in\Delta_3,\;
        \operatorname{dist}_1(q,\mathcal E_3)\ge a_0\delta
    \right\}
    >0.
\]
Consequently,
\[
    \Lambda(p)
    \ge
    \Lambda(H(p))
    \ge
    \gamma(\delta),
\]
proving the claim. Now assume that at least $\kappa n$ rows of $P$ are not $\delta$-two point. Applying the row-slack lower bound to these rows gives
\[
    \sum_{\ell\in L}\Lambda(P_{\ell,\cdot})
    \ge
    \kappa n\,\gamma(\delta).
\]
Plugging this into the row-slack decomposition gives
\[
    H(M)-\Phi(P)
    \le
    n\log\sqrt2-\kappa n\,\gamma(\delta).
\]
Since
\[
    \log\operatorname{per}(A)=H(M)+\mathcal E_A(P)
    \qquad\text{and}\qquad
    \beta(A,P)=\mathcal E_A(P)+\Phi(P),
\]
we have
\[
    \log\operatorname{per}(A)-\beta(A,P)
    =
    H(M)-\Phi(P)
    \le
    n\log\sqrt2-\kappa n\,\gamma(\delta).
\]
Since $P$ is feasible for the Bethe variational problem, $\beta(A,P)\le \log\operatorname{per}_B(A)$, and therefore
\[
    \frac{\operatorname{per}(A)}{\operatorname{per}_B(A)}
    \le
    \left(\sqrt2\,e^{-\kappa\gamma(\delta)}\right)^n
    <
    (\sqrt2-c(\delta,\kappa))^n.
\]
\end{proof}

Note that the above Lemma only gives us that most rows are $\delta$-two point, and nothing about their structure. In order to argue about their structure, we first generalize the machinery of \cite{anari2019tightanalysisbetheapproximation} to an entropy comparison on blocks.

\begin{lemma}
\label{lem:row_block_exposure}
Let $A\in\mathbb R_{\ge0}^{n\times n}$ satisfy $\per(A)>0$, and let $\mu_A$ be the Gibbs measure on perfect matchings. Let $P=(P_{\ell r})$ be the true marginal matrix. For $S\subseteq L$, write $X_S:=(M(\ell))_{\ell\in S}$ and $\Phi_S(P):=\sum_{\ell\in S}\sum_{r\in R}\phi(P_{\ell r})$. Then
\[
    H(M)-\Phi(P)
    \le
    \bigl(H(X_S)-\Phi_S(P)\bigr)
    +
    (n-|S|)\log\sqrt2 .
\]
\end{lemma}

\begin{proof}
Let $T:=L\setminus S$. By the chain rule for entropy,
\[
    H(M)=H(X_S)+H(M\mid X_S).
\]
Since $X_S$ records the partners of all rows in $S$, the conditional randomness in $M$ after $X_S$ is revealed is exactly the randomness of the remaining weighted matching on the rows $T$. Thus
\[
    H(M\mid X_S)
    =
    \mathbb E_x[H(M_T\mid X_S=x)],
\]
where $x$ ranges over feasible values of $X_S$, and $M_T:=(M(\ell))_{\ell\in T}$.

Fix a feasible value $x$ of $X_S$. Let $R_x\subseteq R$ be the set of right vertices not used by the partial matching $x$. Conditional on $X_S=x$, the remaining matching is distributed according to the Gibbs measure of the residual weighted matrix
\[
    A_x:=A[T\times R_x].
\]
Let $Q^x=(Q^x_{\ell r})$ denote the true marginals of this residual Gibbs measure, extended to all $r\in R$ by setting $Q^x_{\ell r}:=0$ for $r\notin R_x$.

Applying Theorem \ref{lem:bethe_perm_known} to the residual matrix $A_x$ and evaluating the Bethe objective at its true marginals gives
\[
    H(M_T\mid X_S=x)
    \le
    |T|\log\sqrt2
    +
    \sum_{\ell\in T}\sum_{r\in R}\phi(Q^x_{\ell r}).
\]
Indeed, the residual energy term cancels from the comparison between the logarithm of the residual partition function and the Bethe objective evaluated at $Q^x$. Taking expectation over $x$, we obtain
\begin{equation}
    H(M\mid X_S)
    \le
    |T|\log\sqrt2
    +
    \mathbb E_x
    \left[
        \sum_{\ell\in T}\sum_{r\in R}\phi(Q^x_{\ell r})
    \right].
\end{equation}

Now, observe that for every $\ell\in T$ and $r\in R$, we have $\mathbb E_x Q^x_{\ell r}=P_{\ell r}$ by the law of total expectation. Additionally, because the function $\psi_d$ is concave on the probability simplex (Lemma \ref{clm:psi-concave-simplex}), we have for each fixed $\ell\in T$,
\[
    \mathbb E_x
    \left[
        \sum_{r\in R}\phi(Q^x_{\ell r})
    \right]
    \le
    \sum_{r\in R}\phi(P_{\ell r}).
\]
Summing over $\ell\in T$, we get
\[
    H(M\mid X_S)
    \le
    |T|\log\sqrt2+
    \sum_{\ell\in T}\sum_{r\in R}\phi(P_{\ell r}).
\]
Combining this with the chain rule and subtracting $\Phi(P)$ gives
\[
    H(M)-\Phi(P)
    \le
    \bigl(H(X_S)-\Phi_S(P)\bigr)
    +(n-|S|)\log\sqrt2.
\]
\end{proof}

Putting together the above two ingredients will enable us to prove the $C_4$ structure of the true marginals.

\begin{lemma}
\label{lem:near_extremal_true_c4_structure_quantitative}
Fix $0<\alpha<1$. There exists $\delta_0>0$ such that for every $0<\delta\le \delta_0$, there exists a constant $c=c(\alpha,\delta)>0$ with the following property.

Let $A\in\mathbb R_{\ge0}^{n\times n}$ satisfy $\per(A)>0$. Let $P$ denote the true marginal matrix of the Gibbs measure $\mu_A$ and define $P_\delta$ to be the subgraph of the support of $A$ consisting of all edges $e$ satisfying $|P_e-1/2|\le \delta$.

If
\[
\frac{\operatorname{per}(A)}{\operatorname{per}_B(A)}>
\bigl(\sqrt2-c\bigr)^n,
\]
then the threshold graph \(P_\delta\) has \(C_4\)-components spanning at least
\((1-\alpha)n\) left vertices.
\end{lemma}

\begin{proof}
Let $M\sim\mu_A$. For $\ell\in L$, write $M(\ell)$ for the right vertex matched to $\ell$. Additionally, let $h_2(p):=-p\log p-(1-p)\log(1-p)$ denote the binary entropy function, and define
\[
    \omega(\delta)
    :=
    \sup_{|x-1/2|\le\delta}|\phi(x)|.
\]
Observe that since $\phi(1/2)=0$, we have $\omega(\delta)\to0$ as $\delta\to0$. Additionally, set $\rho(\delta):= 2h_2(2\delta)+2\delta+2\omega(\delta)$ and choose $\delta_0>0$ sufficiently small so that, for every $0<\delta\le\delta_0$,
\[
    \delta<\frac16
    \qquad\text{and}\qquad
    \gamma_\delta
    :=
    \frac16\log2-\rho(\delta)
    >0.
\]
Fix such a $\delta$, and set
\[
    C_\delta
    :=
    2h_2(2\delta)+2\log2+1+2\omega(\delta).
\]
Additionally, choose $0<\eta<1$ small enough so that $\eta\le \frac{\alpha}{4}$ and $C_\delta\eta\le \frac14\gamma_\delta\alpha$. Let $c_{\mathrm{tp}}=c_{\mathrm{tp}}(\delta,\eta)>0$ be the constant given by Lemma \ref{lem:true_two_point_quantitative}, applied with parameters $\delta$ and $\eta$. Finally define
\[
    c_{\mathrm{str}}
    :=
    \sqrt2\left(1-e^{-\gamma_\delta\alpha/2}\right)>0,
\]
and set $c:=\min\{c_{\mathrm{tp}},c_{\mathrm{str}}\}$.

Now suppose
\[
    \frac{\operatorname{per}(A)}{\operatorname{per}_B(A)}
    >
    \bigl(\sqrt2-c\bigr)^n.
\]
Since $c\le c_{\mathrm{tp}}$, Lemma~\ref{lem:true_two_point_quantitative} implies that all but at most $\eta n$ left vertices are $\delta$-two point. Applying the same lemma to $A^T$ gives the analogous statement for right vertices.

Hence there exist sets $L_{\mathrm{bad}}\subseteq L$ and $R_{\mathrm{bad}}\subseteq R$ such that $|L_{\mathrm{bad}}|\le \eta n$, $|R_{\mathrm{bad}}|\le \eta n$, and every vertex in $(L\setminus L_{\mathrm{bad}})\cup(R\setminus R_{\mathrm{bad}})$ is $\delta$-two point. For every vertex outside these exceptional sets, there are at least two incident edges in $P_\delta$. On the other hand, $P_\delta$ has maximum degree at most $2$. Indeed, if some vertex were incident to three edges of $P_\delta$, then the sum of the corresponding true marginals would be at least $3(1/2-\delta)>1$, a contradiction. Therefore every non-exceptional vertex has exactly two incident edges in $P_\delta$, and every connected component of $P_\delta$ is a path or a cycle.

Let $\mathcal H$ be the collection of connected components of $P_\delta$ which are not $C_4$'s. Enumerate them as $\mathcal H=\{H_1,\dots,H_t\}$. For each $i$, write
\[
    L_i:=L(H_i),
    \qquad
    R_i:=R(H_i),
    \qquad
    F_i:=E(H_i),
\]
and write $S:=\bigcup_{i=1}^t L_i$ and $m:=|S|=\sum_{i=1}^t |L_i|$. We will prove that if $m$ is a positive fraction of $n$, then the Bethe ratio is bounded away from $(\sqrt2)^n$.

We will do this by bounding the entropy of the matching distribution on long cycles and paths. To do this, expose the random vector $X_S:=(M(\ell))_{\ell\in S}$ as follows.
\begin{enumerate}
    \item For each $i\in[t]$, reveal
    \[
        A_i:=\{\ell\in L_i:\ (\ell,M(\ell))\notin F_i\},
    \]
    the set of rows in $H_i$ whose matching edge is not an edge of the $P_\delta$-component $H_i$. Also reveal
    \[
        B_i:=\{r\in R_i:\ (M^{-1}(r),r)\notin F_i\},
    \]
    the set of columns in $H_i$ whose matching edge is not an edge of $H_i$.
    \item For each $i\in[t]$, reveal the actual destinations of the exceptional rows: $Y_i:=(M(\ell))_{\ell\in A_i}$.
    \item Lastly, for each $i\in[t]$, reveal the remaining internal matching inside $H_i$, denoted by $Z_i$.
\end{enumerate}

The collection $((A_i,B_i,Y_i,Z_i))_{i=1}^t$ determines $X_S$. Hence, by the chain rule and entropy subadditivity,
\[
\begin{aligned}
    H(X_S)
    &\le
    H((A_i,B_i)_{i=1}^t)
    +
    H((Y_i)_{i=1}^t\mid (A_i,B_i)_{i=1}^t) \\
    &\qquad
    +
    H((Z_i)_{i=1}^t\mid (A_i,B_i,Y_i)_{i=1}^t).
\end{aligned}
\]

We now bound these three terms. Consider the entropy of the exceptional sets. If $\ell\in L_i\setminus L_{\mathrm{bad}}$, then the two $P_\delta$-edges of $H_i$ incident to $\ell$ both have marginal at least $1/2-\delta$. Therefore
\[
    \mathbb P(\ell\in A_i)
    \le
    1-2\left(\frac12-\delta\right)
    =2\delta.
\]
Similarly, if $r\in R_i\setminus R_{\mathrm{bad}}$, then $\mathbb P(r\in B_i)\le 2\delta$. Thus, by entropy subadditivity,
\[
\begin{aligned}
    H((A_i,B_i)_{i=1}^t)
    &\le
    h_2(2\delta)
    \left(
        \sum_{i=1}^t |L_i|+
        \sum_{i=1}^t |R_i|
    \right)
    +
    \log2(|L_{\mathrm{bad}}|+|R_{\mathrm{bad}}|).
\end{aligned}
\]
Since every component of $P_\delta$ is a path or a cycle, $||L_i|-|R_i||\le1$ for every $i$. If $|L_i|\ne |R_i|$, then $H_i$ is a path whose larger side contains an endpoint. That endpoint has degree at most $1$ in $P_\delta$, and therefore cannot be $\delta$-two point. Hence each imbalanced component can be charged to a vertex of $L_{\mathrm{bad}}\cup R_{\mathrm{bad}}$. Consequently,
\[
    \sum_{i=1}^t ||L_i|-|R_i||
    \le
    |L_{\mathrm{bad}}|+|R_{\mathrm{bad}}|
    \le
    2\eta n.
\]
In particular, $\sum_i |R_i|\le m+2\eta n$, and therefore
\begin{equation} \label{eq:weighted_AiBi_entropy}
    H((A_i,B_i)_{i=1}^t)
    \le
    2m h_2(2\delta)
    +
    (2h_2(2\delta)+2\log2)\eta n.
\end{equation}

Next, we bound the entropy of the exceptional destinations. For $\ell\in L_i$, set
\[
    q_\ell:=\sum_{r:\,(\ell,r)\notin F_i}P_{\ell r}.
\]
Let $W_\ell$ be the random variable equal to $M(\ell)$ if $\ell\in A_i$, and equal to a dummy symbol $*$ otherwise. Then $(Y_i)_{i=1}^t$ is determined by $(W_\ell)_{\ell\in S}$ together with the sets $A_i$. Hence
\[
\begin{aligned}
    H((Y_i)_{i=1}^t\mid (A_i,B_i)_{i=1}^t)
    &\le
    \sum_{i=1}^t\sum_{\ell\in L_i}
        H(W_\ell\mid \mathbf 1_{\{\ell\in A_i\}}) \\
    &\le
    \sum_{i=1}^t
    \sum_{\ell\in L_i}
    \sum_{r:\,(\ell,r)\notin F_i}
        -P_{\ell r}\log P_{\ell r}.
\end{aligned}
\]

Finally, bound the entropy of the internal choices $Z_i$. Once $A_i,B_i,Y_i$ have been revealed, the remaining vertices of $H_i$ must be matched using only edges of $H_i$. If $H_i$ is a path, deleting vertices leaves a disjoint union of paths, and each path has at most one perfect matching. Thus a path component contributes at most one internal choice. If $H_i$ is a cycle, then either some vertex has been deleted, leaving a disjoint union of paths, or no vertex has been deleted, in which case the cycle has exactly two alternating perfect matchings. Therefore a cycle component contributes at most \(\log 2\) to the entropy, while a path
component contributes zero.

Let $c_{\mathrm{cyc}}$ be the number of cycle components in $\mathcal H$. Since the $C_4$-components were excluded, every cycle in $\mathcal H$ has at least $3$ left vertices. Hence $c_{\mathrm{cyc}}\le m/3$. Consequently,
\begin{equation} \label{eq:weighted_internal_entropy}
    H((Z_i)_{i=1}^t\mid (A_i,B_i,Y_i)_{i=1}^t)
    \le
    \frac{m}{3}\log2.
\end{equation}

Combining the three entropy bounds gives
\[
\begin{aligned}
H(X_S)
&\le
2m h_2(2\delta)
+
\sum_{i=1}^t
\sum_{\ell\in L_i}
\sum_{r:\,(\ell,r)\notin F_i}
    -P_{\ell r}\log P_{\ell r}
+
\frac{m}{3}\log 2 \\
&\qquad
+
(2h_2(2\delta)+2\log 2)\eta n .
\end{aligned}
\]

We now compare the entropy with the Bethe contribution from the same rows. Define
\[
\Phi_S(P)=
    \sum_{i=1}^t
    \sum_{\ell\in L_i}
    \sum_{r\in R}
        \phi(P_{\ell r}).
\]
By Lemma \ref{lem:row_block_exposure},
\[
\begin{aligned}
H(M)-\Phi(P)
&\le
(n-m)\log\sqrt2
+
2m h_2(2\delta)
+
\frac{m}{3}\log 2
+
(2h_2(2\delta)+2\log 2)\eta n \\
&\qquad
+
\sum_{i=1}^t
\sum_{\ell\in L_i}
\sum_{r:\,(\ell,r)\notin F_i}
\left(
    -P_{\ell r}\log P_{\ell r}-\phi(P_{\ell r})
\right)
-
\sum_{i=1}^t\sum_{e\in F_i}\phi(P_e) \\
&\le
(n-m)\log\sqrt2
+
2m h_2(2\delta)
+
\frac{m}{3}\log 2
+
(2h_2(2\delta)+2\log 2)\eta n \\
&\qquad
+
\sum_{i=1}^t
\sum_{\ell\in L_i}
\sum_{r:\,(\ell,r)\notin F_i}
    P_{\ell r}
+
\omega(\delta)\sum_{i=1}^t |F_i|,
\end{aligned}
\]
where the last inequality follows from $-P_{\ell r}\log P_{\ell r}-\phi(P_{\ell r})\le P_{\ell r}$ and from $|P_e-1/2|\le\delta$ for every internal edge $e\in F_i$.

Rows in $L_i\setminus L_{\mathrm{bad}}$ contribute at most $2\delta$, while rows in $L_{\mathrm{bad}}$ contribute at most their total row mass. Since $|L_{\mathrm{bad}}|\le \eta n$,
\[
    \sum_{i=1}^t
    \sum_{\ell\in L_i}
    \sum_{r:\,(\ell,r)\notin F_i}
        P_{\ell r}
    \le
    2\delta m+
    \eta n.
\]
Also, $|F_i|\le |L_i|+|R_i|$ and $\sum_i |R_i|\le m+2\eta n$, so
\[
    -\sum_{i=1}^t\sum_{e\in F_i}\phi(P_e)
    \le
    \omega(\delta)(2m+2\eta n).
\]
Plugging these bounds in, we get
\[
\begin{aligned}
H(M)-\Phi(P)
&\le
(n-m)\log\sqrt2
+
\frac{m}{3}\log2
+
(2h_2(2\delta)+2\delta+2\omega(\delta))m \\
&\qquad
+
(2h_2(2\delta)+2\log2+1+2\omega(\delta))\eta n \\
&=
    n\log\sqrt2
    -\gamma_\delta m
    +C_\delta\eta n.
\end{aligned}
\]
Using $\log\operatorname{per}(A)-\beta(A,P)=H(M)-\Phi(P)$ and $\log\operatorname{per}_B(A)\ge \beta(A,P)$, this implies
\begin{equation} \label{eq:weighted_ratio_bound}
    \log\frac{\operatorname{per}(A)}{\operatorname{per}_B(A)}
    \le
    n\log\sqrt2-
    \gamma_\delta m+
    C_\delta\eta n.
\end{equation}

Let $\mathcal C$ be the union of the $C_4$-components of $P_\delta$. Every left vertex not spanned by $\mathcal C$ is either in $S$ or in $L_{\mathrm{bad}}$. Therefore the number of left vertices not spanned by $\mathcal C$ is at most $m+|L_{\mathrm{bad}}|\le m+\eta n$.

Now suppose, toward contradiction, that $\mathcal C$ spans fewer than $(1-\alpha)n$ left vertices. Then $m+\eta n>\alpha n$, and hence, using $\eta\le \alpha/4$, we have $m>(\alpha-\eta)n\ge \frac34\alpha n$. Plugging this into \eqref{eq:weighted_ratio_bound} gives
\[
    \log\frac{\operatorname{per}(A)}{\operatorname{per}_B(A)}
    \le
    n\log\sqrt2
    -
    \frac12\gamma_\delta\alpha n,
\]
where we used $C_\delta\eta\le \gamma_\delta\alpha/4$. Exponentiating, we get
\[
    \frac{\operatorname{per}(A)}{\operatorname{per}_B(A)}
    \le
    \left(\sqrt2\,e^{-\gamma_\delta\alpha/2}\right)^n
    =
    \bigl(\sqrt2-c_{\mathrm{str}}\bigr)^n,
\]
a contradiction since $c\le c_{\mathrm{str}}$.
\end{proof}

Now that we have characterized the behavior of the true marginals under the worst-case ratio, all that remains is showing that this same behavior transfers to weighted low-leakage blocks for the Bethe optimizer and that the corresponding weighted orientation products are nearly balanced. Such a structure is described by the following definition.

\begin{definition} \label{def:true-rho-tight-c4}
Let \(P\) be the true marginal matrix, and let \(P_\rho\) be the threshold graph
consisting of all edges \(e\) with \(|P_e-1/2|\le \rho\). A block
\[
    Q=(\{s_1,s_2\},\{r_1,r_2\})
\]
is called a \emph{true \(\rho\)-tight \(C_4\)-component} if the four edges
$
    (s_1,r_1), (s_1,r_2), (s_2,r_1), (s_2,r_2)
$
form a connected component of \(P_\rho\). In particular, each internal edge
of \(Q\) has true marginal in \([1/2-\rho,1/2+\rho]\), and for every
\(i,j\in\{1,2\}\),
\[
    \sum_{r\notin\{r_1,r_2\}}P_{s_i r}\le 2\rho,
    \qquad
    \sum_{s\notin\{s_1,s_2\}}P_{s r_j}\le 2\rho.
\]
\end{definition}

The following Lemma shows that when $\rho$ is sufficiently small, such tight $C_4$'s can be detected by the Bethe permanent.

\begin{lemma}
\label{lem:true-c4-implies-central-bethe-c4}
For every $0<\lambda<1/10$ and every $\zeta>0$, there exist constants $\rho=\rho(\lambda,\zeta)>0$ and $\varepsilon_{\mathrm{tr}}=\varepsilon_{\mathrm{tr}}(\lambda,\zeta)>0$ such that the following holds.

Let $A\in\mathbb R_{\ge0}^{n\times n}$ satisfy $\operatorname{per}(A)>0$, with true marginal matrix $P$ and Bethe optimizer $P^*$. Suppose
\[
    \log\frac{\operatorname{per}(A)}{\operatorname{per}_B(A)}
    \ge
    n\log\sqrt2-
    \varepsilon_{\mathrm{tr}}n.
\]
Let $\mathcal Q$ be a row-disjoint collection of true $\rho$-tight $C_4$'s for $P$. Then all but at most $\zeta n$ members $Q\in\mathcal Q$ are $\lambda$-admissible for $(A,P^*)$. Moreover every member $Q\in\mathcal Q$ satisfies
\[
    s_Q\ge 2-\lambda.
\]
\end{lemma}

We defer its proof to the next subsection. For now, assuming it is true, we can finish the proof of Theorem \ref{thm:sparse-c4} as follows.

\begin{proof}[Proof of Theorem~\ref{thm:sparse-c4}]
Choose constants $0<\alpha<1/3$ and $0<\zeta<(1-\alpha)/2$. Fix $0<\lambda<1/10$ sufficiently small, to be specified below, and apply Lemma \ref{lem:true-c4-implies-central-bethe-c4} with this choice of parameters $\lambda$ and $\zeta$. Let $\rho_0=\rho_0(\lambda,
\zeta)>0$ and $\varepsilon_{\mathrm{tr}}=\varepsilon_{\mathrm{tr}}(\lambda,\zeta)>0$ be the constants from that lemma. Apply Lemma~\ref{lem:near_extremal_true_c4_structure_quantitative} with the above fixed value of $\alpha$.

Choose $0<\rho\le \rho_0$ small enough so that Lemma~\ref{lem:near_extremal_true_c4_structure_quantitative} applies with tolerance $\rho$ and additionally $2\rho<1/2-\rho_0$. Let $c_{\mathrm{str}}=c_{\mathrm{str}}(\alpha,\rho)>0$ be the corresponding constant. Thus, if
\[
    \frac{\operatorname{per}(A)}{\operatorname{per}_B(A)}
    >
    (\sqrt2-c_{\mathrm{str}})^n,
\]
then the true marginal graph $P_\rho$ contains a vertex-disjoint union of $C_4$'s spanning at least $(1-\alpha)n$ left vertices.

Set
\[
    c_0:=
    \frac14\left(\frac{1-\alpha}{2}-\zeta\right)
    \log(2-\lambda)>0
\]
and choose any $0<\tau<c_0$. Finally choose $\varepsilon>0$ small enough that $0<\varepsilon<c_{\mathrm{str}}$ and $    \log\frac{\sqrt2}{\sqrt2-\varepsilon}
    \le
    \varepsilon_{\mathrm{tr}}.$
We claim that this choice proves the theorem. Suppose, toward contradiction, that $\Gamma_\lambda(A,P^*)<\tau n$ but
\[
    \operatorname{per}(A)>
    (\sqrt2-\varepsilon)^n\operatorname{per}_B(A).
\]
Since $\varepsilon<c_{\mathrm{str}}$, Lemma~\ref{lem:near_extremal_true_c4_structure_quantitative} gives a vertex-disjoint collection $\mathcal Q$ of true $\rho$-tight $C_4$'s spanning at least $(1-\alpha)n$ left vertices, i.e. $    |\mathcal Q|\ge \frac{1-\alpha}{2}n.$
Moreover,
\[
\begin{aligned}
    \log\frac{\operatorname{per}(A)}{\operatorname{per}_B(A)}
    &>
    n\log(\sqrt2-\varepsilon) \\
    &=
    n\log\sqrt2
    -
    n\log\frac{\sqrt2}{\sqrt2-\varepsilon} \\
    &\ge
    n\log\sqrt2-
    \varepsilon_{\mathrm{tr}}n.
\end{aligned}
\]
We claim that every true $\rho$-tight $C_4$-component is also a true $\rho_0$-tight $C_4$-component. Indeed, the four internal edges are within $\rho\le \rho_0$ of $1/2$. Moreover, for every vertex of the block, the total marginal mass leaving the block is at most $2\rho$. Hence every external incident edge has marginal at most $2\rho<1/2-\rho_0$, and therefore no external incident edge belongs to the threshold graph $P_{\rho_0}$.

Thus Lemma~\ref{lem:true-c4-implies-central-bethe-c4} applies to $\mathcal Q$. All but at most $\zeta n$ members of $\mathcal Q$ are $\lambda$-admissible for $(A,P^*)$, and every member of $\mathcal Q$ has score at least $\log(2-\lambda)$. Therefore the maximum total score of a vertex-disjoint collection of $\lambda$-admissible blocks is at least
\[
    \left(\frac{1-\alpha}{2}-\zeta\right)n\log(2-\lambda).
\]
The greedy algorithm that defines $\Gamma_\lambda(A,P^*)$ obtains at least a $1/4$ fraction of the optimum score: each selected block uses four vertices, and a vertex-disjoint optimum block can be charged to the first greedy block with which it intersects. Hence
\[
    \Gamma_\lambda(A,P^*)
    \ge
    c_0 n
    >
    \tau n,
\]
contradicting the assumption $\Gamma_\lambda(A,P^*)<\tau n$.
\end{proof}

\subsection{Proof of Lemma \ref{lem:true-c4-implies-central-bethe-c4}}

Recall the notation
\[
\phi(x):=-x\log x+(1-x)\log(1-x),
\qquad
\psi_d(x_1,\dots,x_d):=\sum_{i=1}^d \phi(x_i).
\]
When the dimension is clear, we simply write $\psi$, and recall that by Lemma \ref{clm:psi-concave-simplex}, it is concave.

\begin{claim}
\label{clm:merging-does-not-increase-gap}
Define $T:\Delta_d\to\Delta_3$ by
\[
T(x_1,\dots,x_d)=
\left(x_1,x_2,\sum_{i=3}^d x_i\right).
\]
Then, for all $p,q\in\Delta_d$,
\[
\psi_d\!\left(\frac{p+q}{2}\right)
-\frac12\psi_d(p)
-\frac12\psi_d(q)
\ge
\psi_3\!\left(\frac{Tp+Tq}{2}\right)
-\frac12\psi_3(Tp)
-\frac12\psi_3(Tq).
\]
\end{claim}

\begin{proof}
It suffices to prove that merging two coordinates cannot increase the Jensen gap, and the result follows by iterating. Denote the two coordinates being merged in $p$ by $u,v$, and in $q$ the corresponding coordinates by $u',v'$. The loss in $\psi$ from merging these two coordinates is
\[
D(u,v):=\phi(u)+\phi(v)-\phi(u+v)
=
\phi(u)+\phi(v)+\phi(1-u-v)
=
\psi_3(u,v,1-u-v).
\]
Thus $D$ is concave on the triangle $\{(u,v):u\ge0,\ v\ge0,\ u+v\le1\}$ by Lemma \ref{clm:psi-concave-simplex}. Therefore
\[
D\!\left(\frac{u+u'}2,\frac{v+v'}2\right)
\ge
\frac12D(u,v)+\frac12D(u',v').
\]
Equivalently, the Jensen gap before merging is at least the Jensen gap after merging. Iterating gives the claim.
\end{proof}

The next claim also exploits the concavity of the Bethe objective on the simplex.

\begin{claim}
\label{clm:true-marginals-near-bethe-quantitative}
Let $P^*$ be a Bethe optimizer and let $P$ be the true marginal matrix. Suppose that, for some $\varepsilon_{\rm tr}>0$,
\[
\log\frac{\operatorname{per}(A)}{\operatorname{per}_B(A)}
\ge
n\log\sqrt2-
\varepsilon_{\rm tr}n.
\]
Then
\begin{enumerate}
    \item
    $
    0
    \le
    \log \operatorname{per}_B(A;P^*)
    -
    \log \operatorname{per}_B(A;P)
    \le
    \varepsilon_{\rm tr} n.
    $
    \item
    $
    \sum_{\ell\in L}
    \left[
    \psi\!\left(\frac{P_{\ell,\cdot}+P^*_{\ell,\cdot}}{2}\right)
    -
    \frac12 \psi(P_{\ell,\cdot})
    -
    \frac12 \psi(P^*_{\ell,\cdot})
    \right]
    \le
    \frac{\varepsilon_{\rm tr}}{2}n .
    $
\end{enumerate}
\end{claim}

\begin{proof}
Since $P^*$ is a Bethe optimizer,
\[
\operatorname{per}_B(A;P)\le \operatorname{per}_B(A;P^*)=\operatorname{per}_B(A).
\]
Therefore the difference in the first item is nonnegative. Also, by Theorem \ref{lem:bethe_perm_known},
\[
\log \operatorname{per}(A)
\le
n\log\sqrt2+
\log\operatorname{per}_B(A;P).
\]
On the other hand, the hypothesis gives
\[
\log \operatorname{per}(A)
\ge
n\log\sqrt2-
\varepsilon_{\rm tr}n+
\log\operatorname{per}_B(A).
\]
Combining the previous two displays proves the first item.

For the second, set $\overline P=(P+P^*)/2$. By concavity of $X\mapsto \log\operatorname{per}_B(A;X)$ on the matching polytope,
\[
\log\operatorname{per}_B(A;\overline P)
\ge
\frac12\log\operatorname{per}_B(A;P)
+
\frac12\log\operatorname{per}_B(A;P^*).
\]
On the other hand, by optimality of $P^*$,
\[
\log\operatorname{per}_B(A;\overline P)
\le
\log\operatorname{per}_B(A;P^*).
\]
Thus
\[
\begin{aligned}
0
&\le
\log\operatorname{per}_B(A;\overline P)
-
\frac12\log\operatorname{per}_B(A;P)
-
\frac12\log\operatorname{per}_B(A;P^*) \\
&\le
\frac12
\left(
\log\operatorname{per}_B(A;P^*)-
\log\operatorname{per}_B(A;P)
\right)
\le
\frac{\varepsilon_{\rm tr}}{2}n.
\end{aligned}
\]
The linear energy term $\mathcal E_A$ cancels from the Jensen gap, so the left-hand side equals
\[
\sum_{\ell\in L}
\left[
\psi\!\left(\frac{P_{\ell,\cdot}+P^*_{\ell,\cdot}}{2}\right)
-
\frac12 \psi(P_{\ell,\cdot})
-
\frac12 \psi(P^*_{\ell,\cdot})
\right].
\]
This proves the claim.
\end{proof}

\begin{claim}
\label{claim:row-leakage-jensen-gap}
For every $\alpha>0$, there exist constants $\rho_{\rm row}=\rho_{\rm row}(\alpha)>0$ and $\kappa=\kappa(\alpha)>0$ such that the following holds.

Let $p,q\in\Delta_d$. Suppose that $p$ is $\rho_{\rm row}$-tight on coordinates $1$ and $2$, i.e.
\[
\left|p_1-\frac12\right|\le \rho_{\rm row},
\qquad
\left|p_2-\frac12\right|\le \rho_{\rm row},
\qquad
\sum_{i\ge3}p_i\le \rho_{\rm row},
\]
and suppose that $q$ satisfies $\sum_{i\ge3}q_i\ge \alpha$. Then
\[
\psi_d\!\left(\frac{p+q}{2}\right)
-
\frac12\psi_d(p)
-
\frac12\psi_d(q)
\ge \kappa .
\]
\end{claim}

\begin{proof}
Define $p^*=(1/2,1/2,0)$. Since $p$ is $\rho_{\rm row}$-tight,
\[
\|Tp-p^*\|_1
\le
\left|p_1-\frac12\right|
+
\left|p_2-\frac12\right|
+
\sum_{i\ge3}p_i
\le 3\rho_{\rm row}.
\]
For $u,v\in\Delta_3$, define
\[
J(u,v)
:=
\psi_3\!\left(\frac{u+v}{2}\right)
-
\frac12\psi_3(u)
-
\frac12\psi_3(v).
\]
We show that $J(p^*,v)$ is bounded below by a positive constant whenever $v_3\ge\alpha$. Let $v=(a,b,c)\in\Delta_3$ with $c\ge\alpha$, and define
\[
    f(t):=\psi_3((1-t)p^*+tv),
    \qquad t\in[0,1].
\]
Then
\[
    f(t)
    =
    \phi\!\left(\frac{1-t}{2}+ta\right)
    +
    \phi\!\left(\frac{1-t}{2}+tb\right)
    +
    \phi(tc).
\]
By Lemma \ref{clm:psi-concave-simplex}, $f$ is concave on $[0,1]$. It is not affine, because $\phi(0)=0$ and, as $x\downarrow0$,
\[
    \frac{\phi(x)-\phi(0)}{x}
    \longrightarrow +\infty.
\]
Since $c\ge\alpha>0$, the third term $\phi(tc)$ gives $\lim_{t\downarrow0}(f(t)-f(0))/t=+\infty$. Therefore
\[
    f\!\left(\frac12\right)>
    \frac12f(0)+\frac12f(1).
\]
Equivalently, $J(p^*,v)>0$ whenever $v_3\ge\alpha$. The set $K_\alpha:=\{v\in\Delta_3:v_3\ge\alpha\}$ is compact, and $v\mapsto J(p^*,v)$ is continuous. Hence
\[
\kappa_0(\alpha):=
\min_{v\in K_\alpha}J(p^*,v)>0.
\]
By continuity of $J$ on $\Delta_3\times K_\alpha$, after decreasing $\rho_{\rm row}(\alpha)>0$ if necessary, we may ensure that whenever $\|u-p^*\|_1\le 3\rho_{\rm row}$ and $v\in K_\alpha$, one has $J(u,v)\ge \kappa_0(\alpha)/2$. Set $\kappa(\alpha):=\kappa_0(\alpha)/2$. Now take $u=Tp$ and $v=Tq$. Since $p$ is $\rho_{\rm row}$-tight and $\sum_{i\ge3}q_i\ge\alpha$, the result follows from Claim \ref{clm:merging-does-not-increase-gap}.
\end{proof}

The next claim says that a true-tight $C_4$ has nearly balanced orientation products.

\begin{claim}
\label{clm:true_tight_balanced_weights}
There is an absolute constant $C>0$ such that the following holds. Let $Q=(\{s_1,s_2\},\{r_1,r_2\})$ be a true $\rho$-tight $C_4$-component for the Gibbs measure of $A$, with $\rho<1/20$. Then
\[
    \frac{\min\{A_{s_1r_1}A_{s_2r_2},A_{s_1r_2}A_{s_2r_1}\}}
    {\max\{A_{s_1r_1}A_{s_2r_2},A_{s_1r_2}A_{s_2r_1}\}}
    \ge
    1-C\rho.
\]
In particular $s_Q\ge 2-C\rho$.
\end{claim}

\begin{proof}
Let
\[
D:=\Pr_{\mu_A}[M(s_1)=r_1,\ M(s_2)=r_2]
\]
and
\[
X:=\Pr_{\mu_A}[M(s_1)=r_2,\ M(s_2)=r_1].
\]
Because $Q$ is true $\rho$-tight, each internal edge marginal is within $\rho$ of $1/2$, and each row and column leakage from $Q$ is at most $2\rho$. If the event $M(s_1)=r_1$ occurs but $M(s_2)\ne r_2$, then column $r_2$ is matched from outside $\{s_1,s_2\}$, an event of probability at most $2\rho$. Therefore
\[
    D\ge P_{s_1r_1}-2\rho\ge \frac12-3\rho.
\]
Also $D\le P_{s_1r_1}\le 1/2+\rho$. The same argument gives
\[
    \frac12-3\rho\le X\le \frac12+\rho.
\]

Conditioned on $Q$ being internally closed, the residual outside problem is identical for the parallel and cross orientations. Hence
\[
    \frac{D}{X}
    =
    \frac{A_{s_1r_1}A_{s_2r_2}}{A_{s_1r_2}A_{s_2r_1}}.
\]
Since both $D$ and $X$ lie in $[1/2-3\rho,1/2+\rho]$, their ratio is between $1-C\rho$ and $(1-C\rho)^{-1}$ for an absolute constant $C$, proving the claim.
\end{proof}

With the above claims in hand, we now prove Lemma \ref{lem:true-c4-implies-central-bethe-c4}.

\begin{proof}[Proof of Lemma \ref{lem:true-c4-implies-central-bethe-c4}]
Fix $0<\lambda<1/10$ and $\zeta>0$. Choose $    \alpha:=\frac{\lambda}{4}.$ Let $\rho_{\rm row}=\rho_{\rm row}(\alpha)$ and $\kappa=\kappa(\alpha)$ be the constants from Claim \ref{claim:row-leakage-jensen-gap}. Choose $\rho>0$ small enough so that
\[
    \rho\le \frac12\rho_{\rm row},
    \qquad
    \rho<\frac1{20},
    \qquad
    C\rho\le \lambda,
\]
where $C$ is the constant from Claim \ref{clm:true_tight_balanced_weights}. Lastly, define $ \varepsilon_{\rm tr}:=\kappa\zeta.$ Let $\mathcal Q$ be a row-disjoint collection of true $\rho$-tight $C_4$'s for the true marginal matrix $P$, and suppose
\[
\log\frac{\operatorname{per}(A)}{\operatorname{per}_B(A)}
\ge
n\log\sqrt2-
\varepsilon_{\rm tr}n.
\]
Call a row $s_i$ bad if it belongs to some $Q=\{s_1,s_2,r_1,r_2\}\in\mathcal Q$ and satisfies
\[
    \sum_{r\notin\{r_1,r_2\}}P^*_{s_i r}\ge \alpha.
\]
We show that there are few bad rows. Fix a bad row $s_i$ belonging to such a $Q\in\mathcal Q$. Since $Q$ is a true $\rho$-tight $C_4$-component for $P$, the two internal true marginals in row $s_i$ are within $\rho$ of $1/2$, and
\[
    \sum_{r\notin\{r_1,r_2\}}P_{s_i r}
    \le
    2\rho
    \le
    \rho_{\rm row}.
\]
Also $\rho\le\rho_{\rm row}$, so after relabeling $r_1,r_2$ as coordinates $1,2$, the row $P_{s_i,\cdot}$ is $\rho_{\rm row}$-tight in the sense of Claim \ref{claim:row-leakage-jensen-gap}. Since $s_i$ is bad,
\[
    \sum_{r\notin\{r_1,r_2\}}P^*_{s_i r}
    \ge
    \alpha.
\]
Therefore Claim \ref{claim:row-leakage-jensen-gap} implies that this row contributes at least $\kappa$ to the Jensen gap:
\[
\psi\!\left(\frac{P_{s_i,\cdot}+P^*_{s_i,\cdot}}2\right)
-
\frac12\psi(P_{s_i,\cdot})
-
\frac12\psi(P^*_{s_i,\cdot})
\ge \kappa.
\]
By Lemma \ref{clm:psi-concave-simplex}, every row Jensen gap is nonnegative. On the other hand, by Claim \ref{clm:true-marginals-near-bethe-quantitative}, the total row Jensen gap is at most $\varepsilon_{\rm tr}n/2$. Thus the number of bad rows is at most
\[
    \frac{\varepsilon_{\rm tr}}{2\kappa}n
    =
    \frac{\zeta}{2}n.
\]
Since $\mathcal Q$ is row-disjoint, the number of blocks in $\mathcal Q$ containing at least one bad row is also at most $\zeta n/2$. In particular, all but at most $\zeta n/2$ blocks in $\mathcal Q$ have both Bethe row leakages less than $\alpha$.

Now fix such a good block $Q=\{s_1,s_2,r_1,r_2\}\in\mathcal Q$. Then
\[
    \sum_{r\notin\{r_1,r_2\}}P^*_{s_i r}<\alpha
    \qquad
    (i=1,2).
\]
We next show that the two Bethe column leakages are also small. Since $P^*$ is feasible, all row and column sums equal $1$. Therefore
\[
\begin{aligned}
\sum_{s\notin\{s_1,s_2\}}P^*_{s r_1}
+
\sum_{s\notin\{s_1,s_2\}}P^*_{s r_2}
&=
2-
\sum_{i=1}^2\sum_{j=1}^2P^*_{s_i r_j} \\
&=
\sum_{i=1}^2
\left(1-
\sum_{j=1}^2P^*_{s_i r_j}
\right) \\
&=
\sum_{i=1}^2\sum_{r\notin\{r_1,r_2\}}P^*_{s_i r}
<2\alpha.
\end{aligned}
\]
Thus the total row leakage is less than $2\alpha$ and the total column leakage is less than $2\alpha$. Hence
\[
    \operatorname{leak}_{P^*}(Q)<4\alpha=\lambda.
\]
So all but at most $\zeta n$ members of $\mathcal Q$ are $\lambda$-admissible. Finally, by Claim \ref{clm:true_tight_balanced_weights} and the choice of $\rho$, every member $Q\in\mathcal Q$ satisfies $s_Q\ge2-\lambda$. This proves the lemma.
\end{proof}

\end{document}